\documentclass[submission,copyright,creativecommons]{eptcs}
\providecommand{\event}{NCMA 2023} 

\usepackage{iftex}
\usepackage{combelow}

\ifpdf
  \usepackage{underscore}         
  \usepackage[T1]{fontenc}        
\else
  \usepackage{breakurl}           
\fi

\usepackage{extarrows,paralist,amssymb,amsmath,graphicx}
\usepackage{tikz,pgf}\usetikzlibrary{positioning,arrows,automata,backgrounds}

\title{Merging two Hierarchies of Internal Contextual Grammars with Subregular Selection}
\author{Bianca Truthe 
\institute{Institut f\"ur Informatik, Universit\"at Giessen, Arndtstr.~2, 35392 Giessen, Germany}
\email{bianca.truthe@informatik.uni-giessen.de}}
\def\titlerunning{Merging two Hierarchies of Internal Contextual Grammars with Subregular Selection}
\def\authorrunning{B.~Truthe}

\def\Set#1#2{\left\{\: #1\;|\; #2\:\right\}}

\def\Sets#1{\left\{\,#1\,\right\}}
\def\set#1#2{\{\; #1 \mid #2\;\}}
\def\sets#1{\{#1\}}

\newcommand{\CIRC}{\mathit{CIRC}}
\newcommand{\COMB}{\mathit{COMB}}
\newcommand{\COMM}{\mathit{COMM}}
\newcommand{\DEF}{\mathit{DEF}}
\newcommand{\FIN}{\mathit{FIN}}
\newcommand{\MON}{\mathit{MON}}
\newcommand{\NC}{\mathit{NC}}
\newcommand{\NIL}{\mathit{NIL}}
\newcommand{\ORD}{\mathit{ORD}}
\newcommand{\PS}{\mathit{PS}}
\newcommand{\REG}{\mathit{REG}}
\newcommand{\SUF}{\mathit{SUF}}
\newcommand{\UF}{\mathit{UF}}
\newcommand{\RL}{\mathit{RL}}

\newcommand{\Suf}{\mathit{Suf}}

\def\cF{{\cal F}}

\def\cS{{\cal S}}

\def\cIC{{\cal IC}}

\def\Lra{\Longrightarrow}
\def\ra{\rightarrow}

\tikzstyle{to}=[->, >=stealth]
\tikzstyle{hier}=[->, >=angle 60]
\tikzstyle{hiero}=[->, >=angle 60, dashed]
\tikzstyle{state}=[circle,draw,inner sep=2pt,minimum size=8mm]
\tikzstyle{edgeLabel}=[inner sep=0.5mm,fill=white,text=black!60]

\newtheorem{theorem}{Theorem}[section]

\newtheorem{lemma}[theorem]{Lemma}

\newenvironment{proof}{{\em Proof. }}{{}\hspace*{\fill}$\Box$ \par \medskip }
\newenvironment{proof*}{{\em Proof. }}{\par \medskip }

\newlength{\btlabelwidth}
\newlength{\btleftmargin}

\definecolor{lightred}{rgb}{1,0.5,0.5}

\begin{document}
\maketitle
\begin{abstract}
In this paper, we continue the research on the power of contextual grammars with selection languages from 
subfamilies of the family of regular languages. In the past, two independent hierarchies have been obtained 
for external and internal contextual grammars, one based on selection languages defined by structural properties 
(finite, monoidal, nilpotent, combinational, definite, ordered, non-counting,  power-sepa\-rating,
suffix-closed, commutative, circular, or union-free languages), 
the other one based on selection languages defined by resources (number of non-terminal symbols, 
production rules, or states needed for generating or accepting them). In a previous paper,
the language families of these hierarchies for external contextual grammars were compared and the hierarchies merged.
In the present paper, we compare the language families of these hierarchies for internal contextual grammars 
and merge these hierarchies. 
\end{abstract}

\section{Introduction}
Contextual grammars were introduced by S.~Marcus in \cite{Mar69} as a formal model that might be 
used in the generation of natural languages. The derivation steps consist in adding contexts 
to given well formed sentences, starting from an initial finite basis. Formally, a context 
is given by a pair $(u,v)$ of words and 
inserting it externally into a word $x$ gives the word $uxv$ 
whereas inserting it internally gives all words~$x_1ux_2vx_3$ when $x=x_1x_2x_3$. In order to
control the derivation process, contextual grammars with selection 
were defined. In such contextual grammars, a context $(u,v)$ may be added only if the 
surrounded word $x$ or $x_2$ belongs to a language which is associated with the context.
Language families were defined where all selection languages in a contextual grammar belong
to some language family $\cF$. Such contextual grammars are said to be `with selection in the family $\cF$'.
Contextual grammars have been studied where
the family $\cF$ is taken from the Chomsky hierarchy (see~\cite{Ist78,Pau98,handbook} and references therein).

In \cite{Das05}, the study of external contextual grammars with selection in special regular sets 
was started. Finite, combinational, definite, nilpotent, regular suffix-closed, regular commutative 
languages and languages of the form $V^*$ for some alphabet $V$ were considered. 
The research was continued in \cite{DasManTru11b,DasManTru12a,DasManTru12b,ManTru12}
where further subregular families of selection languages were considered and the effect of 
subregular selection languages on the generative power of external and internal contextual grammars 
was investigated. A recent survey can be found in \cite{Tru21-fi} which presents for each type of 
contextual grammars (external and internal ones) two hierarchies, one based on selection languages
defined by structural properties (finite, monoidal, nilpotent, combinational, definite, ordered, 
non-counting,  power-separating, suffix-closed, commutative, circular, or union-free languages), 
the other one based on selection languages defined by resources (number of non-terminal symbols, 
production rules, or states needed for generating or accepting them). In \cite{Tru23-dcfs},
the language families of these hierarchies for external contextual grammars were compared and the hierarchies merged.
In the present paper, we 
compare the language families of these hierarchies for internal contextual grammars and merge the hierarchies.

The internal case is different from the case of external contextual grammars, as there are two 
main differences between the ways in which words are derived. In the case of internal contextual 
grammars, it is possible that the insertion of a context into a sentential form can be done at 
more than one place, such that the derivation becomes in some sense non-deterministic; in the
case of external grammars, once a context was selected, there is at most one way to insert it: 
wrapped around the sentential form, when this word is in the selection language of the context. 
On the other hand, the outermost ends of a word derived externally have been added at the end 
of the derivation, whereas derived internally the ends could have been at the ends of the word 
already from the beginning since some inner part can be `pumped'. 
If a context can be added internally, then it can be added arbitrarily often (because the 
subword where the context is wrapped around does not change) which does not necessarily hold
for external grammars.

In Section 2, we give the definitions and notation of the concepts used in this paper (languages,
grammars, automata, 
subregular language families, inclusion relations between these 
families, contextual grammars, and inclusion relations between the families generated by internal contextual
grammars where the selection languages belong to various subregular language families).
In Section 3, we present our results where, first, several languages are presented which later serve
as witness languages for proper inclusions or the incomparability of two language families and, later,
these languages are used to prove relations between the various language families generated by
internal contextual grammars with different types of selection.
Finally, in Section 4, we state some problems which are left open and give some ideas for future research.

\section{Preliminaries}
Throughout the paper, we assume that the reader is familiar with the basic concepts of the 
theory of automata and formal languages. 
For details, we refer to~\cite{handbook}. Here we only recall some notation and the definition 
of contextual grammars with selection which form the central notion of the paper.

\subsection{Languages, grammars, automata}

Given an alphabet $V$, we denote by $V^*$ and $V^+$ the set of all words and the set of all non-empty words over $V$,
respectively. The empty word is denoted by~$\lambda$. By $V^k$
and $V^{{}\leq k}$ for some natural number $k$,
we denote the set of all words of the alphabet $V$ with exactly $k$ letters
and the set of all words over $V$ with at most $k$ letters, respectively.
For a word $w$ and a letter $a$, we denote the length of $w$ by $|w|$ and the number of occurrences of the letter $a$ in the
word~$w$ by $|w|_a$. For a set $A$, we denote its cardinality by $|A|$.

A right-linear grammar is a quadruple 
\[G=(N,T,P,S)\]
where $N$ is a finite set of non-terminal symbols, $T$ is a finite set of
terminal symbols, $P$ is a finite set of production rules of the form $A\ra wB$ or $A\ra w$ with $A,B\in N$ and~$w\in T^*$,
and $S\in N$ is the start symbol. Such a grammar is called regular, if all the rules are of the form $A\ra xB$
or $A\ra x$ with $A,B\in N$ and~$x\in T$ or $S\ra\lambda$. The language generated by a right-linear or regular grammar
is the set of all words over the terminal alphabet which are obtained from the start symbol $S$ by a successive replacement
of the non-terminal symbols according to the rules in the set $P$. Every language generated by a right-linear grammar
can also be generated by a regular grammar.

A deterministic finite automaton is a quintuple 
\[{\cal A}=(V,Z,z_0,F,\delta)\]
where $V$ is a finite set of input symbols, $Z$
is a finite set of states, $z_0\in Z$ is the initial state, $F\subseteq Z$ is a set of accepting states, and $\delta$ is
a transition function $\delta: Z\times V\to Z$. The language accepted by such an automaton is the set of all input words 
over the alphabet $V$ which lead letterwise by the transition function from the initial state to an accepting state.

The set of all languages generated by some right-linear grammar coincides with the set of all languages accepted by a
deterministic finite automaton. All these languages are called regular and form a family denoted by $\REG$. Any subfamily
of this set is called a subregular language family.



\subsection{Resources restricted languages}

We define subregular families by restricting the resources needed
for generating or accepting their elements:
\begin{align*}
\RL_n^V &= \Set{ L }{ \text{$L$ is generated by a right-linear grammar with at most $n$ non-terminal symbols} },\\
\RL_n^P &= \Set{ L }{ \text{$L$ is generated by a right-linear grammar with at most $n$ production rules} },\\
\REG_n^Z &= \Set{ L }{ \text{$L$ is accepted by a deterministic finite automaton with at most $n$ states} }.
\end{align*}

\subsection{Subregular language families based on the structure}

We consider the following restrictions for regular languages. Let $L$ be a language
over an alphabet $V$. 
With respect to the alphabet $V$, the language $L$ is said to be
\begin{itemize}
\item \emph{monoidal} if and only if $L=V^*$,
\item \emph{nilpotent} if and only if it is finite or its complement $V^*\setminus L$ is finite,
\item \emph{combinational} if and only if it has the form
$L=V^*X$
for some subset $X\subseteq V$,
\item \emph{definite} if and only if it can be represented in the form
$L=A\cup V^*B$
where~$A$ and~$B$ are finite subsets of $V^*$,
\item \emph{suffix-closed} (or \emph{fully initial} or \emph{multiple-entry} language) if
and only if, for any two words~$x\in V^*$ and~$y\in V^*$, the relation $xy\in L$ implies
the relation~$y\in L$,
\item \emph{ordered} if and only if the language is accepted by some deterministic finite
automaton 
\[{\cal A}=(V,Z,z_0,F,\delta)\]
with an input alphabet $V$, a finite set $Z$ of states, a start state $z_0\in Z$, a set $F\subseteq Z$ of
accepting states and a transition mapping $\delta$ where $(Z,\preceq )$ is a totally ordered set and, for
any input symbol~$a\in V$, the relation $z\preceq z'$ implies $\delta (z,a)\preceq \delta (z',a)$,
\item \emph{commutative} if and only if it contains with each word also all permutations of this
word,
\item \emph{circular} if and only if it contains with each word also all circular shifts of this
word,
\item \emph{non-counting} (or \emph{star-free}) if and only if there is a natural
number $k\geq 1$ such that, for any three words $x\in V^*$, $y\in V^*$, and $z\in V^*$, it 
holds~$xy^kz\in L$ if and only if $xy^{k+1}z\in L$,
\item \emph{power-separating} if and only if, there is a natural number $m\geq 1$ such that
for any word~$x\in V^*$, either
$J_x^m \cap L = \emptyset$
or
$J_x^m\subseteq L$
where
$J_x^m = \set{ x^n}{n\geq m}$,
\item \emph{union-free} if and only if $L$ can be described by a regular expression which
is only built by product and star.
\end{itemize}

We remark that monoidal, nilpotent, combinational, definite, ordered, and union-free 
languages are regular, whereas non-regular languages of the other types mentioned above exist.
Here, we consider among the commutative, circular, suffix-closed, non-counting,
and power-separating languages only those which are also regular.

Some properties of the languages of the classes mentioned above can be found in
\cite{Shyr91} (monoids), \cite{GecsegPeak72} (nilpotent languages), \cite{Ha69} (combinational and commutative languages),
\cite{PerRabSham63} (definite languages), \cite{GilKou74} and \cite{BrzoJirZou14} (suffix-closed languages),
\cite{ShyThi74-ORD} (ordered languages), \cite{Das79} (circular languages),
\cite{McNPap71} (non-counting languages), \cite{ShyThi74-PS} (power-separating languages),
\cite{Brzo62} (union-free languages).

By $\FIN$, $\MON$, $\NIL$, $\COMB$, $\DEF$, $\SUF$, $\ORD$, $\COMM$, $\CIRC$, $\NC$, $\PS$, $\UF$, 
and $\REG$, we denote the families of all finite, 
monoidal, nilpotent, combinational, definite, regular suffix-closed, ordered, regular commutative, regular circular, 
regular non-counting, regular power-separating, union-free, 
and regular, languages, respectively.


As the set of all families under consideration, we set
\begin{align*}
\mathfrak{F} &= \sets{ \FIN, \MON, \NIL, \COMB, \DEF, \SUF, \ORD, \COMM, \CIRC, \NC, \PS, \UF}\\
    &\quad{}\cup\set{ \RL_n^V}{n\geq 1}\cup\set{\RL_n^P}{n\geq 1}\cup\set{\REG_n^Z}{n\geq 1}.
\end{align*}

\subsection{Hierarchy of subregular families of languages}

We present here a hierarchy of the families of the aforementioned set $\mathfrak{F}$ with respect to the
set theoretic inclusion relation.

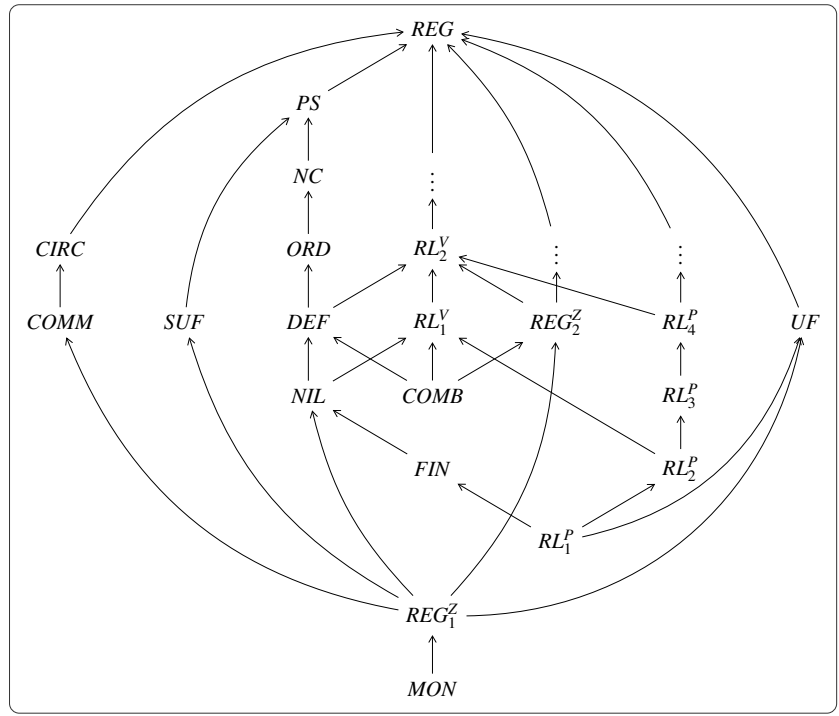
\begin{figure}[htb]
\centerline{%
\scalebox{0.75}{\begin{tikzpicture}[node distance=13mm and 22mm,on grid=true,
background rectangle/.style=
{
draw=black!80,
rounded corners=1ex},
show background rectangle]
\node (REG) {$\REG$};
\node (PS) [below left=of REG] {$\PS$};
\node (NC) [below=of PS] {$\NC$};
\node (ORD) [below=of NC] {$\ORD$};
\node (DEF) [below=of ORD] {$\DEF$};
\node (COMB) [below right=of DEF] {$\COMB$};
\node (NIL) [below=of DEF] {$\NIL$};
\node (FIN) [below right=of NIL] {$\FIN$};
\node (SUF) [left=of DEF] {$\SUF$};
\node (COMM) [left=of SUF] {$\COMM$};
\node (CIRC) [above=of COMM] {$\CIRC$};
\node (V1) [right=of DEF] {$\RL_1^V$};
\node (V2) [above=of V1] {$\RL_2^V$};
\node (Vn) [right=of NC] {$\vdots$};
\node (Z2) [right=of V1] {$\REG_2^Z$};
\node (Zn) [right=of V2] {$\vdots$};
\node (dummy) [right=of Z2] {};
\node (d1) [below=of dummy] {};
\node (UF) [right=of dummy] {$\UF$};
\node (d2) [below=of d1] {};
\node (P1) [below left=of d2] {$\RL_1^P$};
\node (P2) [above right=of P1] {$\RL_2^P$};
\node (P3) [above=of P2] {$\RL_3^P$};
\node (P4) [above=of P3] {$\RL_4^P$};
\node (Pn) [above=of P4] {$\vdots$};
\node (d3) [below=of FIN] {};
\node (Z1) [below=of d3] {$\REG_1^Z$};
\node (MON) [below=of Z1] {$\MON$};
\draw[hier] (FIN) to 
  (NIL);
\draw[hier] (MON) to 
  (Z1);
\draw[hier] (Z1) [bend right=-15] to 
  (NIL);
\draw[hier] (Z1) [bend right=-22] to 
  (SUF);
\draw[hier] (Z1) [bend right=-29] to 
  (COMM);
\draw[hier] (Z1) [bend right=39] to 
  (UF);
\draw[hier] (Z1) [bend right=20] to 
  (Z2);
\draw[hier] (P1) to 
  (FIN);
\draw[hier] (P1) [bend right=30] to 
  (UF);
\draw[hier] (V1) to 
  (V2);
\draw[hier] (V2) to 
  (Vn);
\draw[hier] (Vn) to 
  (REG);
\draw[hier] (Z2) to 
  (V2);
\draw[hier] (Z2) to 
  (Zn);
\draw[hier] (Zn) [bend right=15] to 
  (REG);
\draw[hier] (P1) to 
  (P2);
\draw[hier] (P2) to 
  (P3);
\draw[hier] (P2) 
  to 
  (V1);
\draw[hier] (P3) to 
  (P4);
\draw[hier] (P4) to 
  (Pn);
\draw[hier] (P4) 
  to 
  (V2);
\draw[hier] (Pn) [bend right=22] to 
  (REG);
\draw[hier] (NIL) to 
  (DEF);
\draw[hier] (NIL) 
  to 
  (V1);
\draw[hier] (COMB) 
  to 
  (DEF);
\draw[hier] (COMB) to 
  (V1);
\draw[hier] (COMB) to 
  (Z2);
\draw[hier] (ORD) to 
  (NC);
\draw[hier] (DEF) to 
  (ORD);
\draw[hier] (DEF) to 
  (V2);
\draw[hier] (NC) to 
  (PS);
\draw[hier] (PS) to 
  (REG);
\draw[hier] (SUF) [bend right=-22] to 
  (PS);
\draw[hier] (COMM) to 
  (CIRC);
\draw[hier] (CIRC) [bend right=-25] to 
  (REG);
\draw[hier] (UF) [bend right=29] to 
  (REG);
\end{tikzpicture}}}
\caption{Hierarchy of subregular language families}\label{fig-subreg-hier}
\end{figure}

\begin{theorem}\label{th-subreg-hier}
The inclusion relations presented in Figure~\ref{fig-subreg-hier} hold.
An arrow from an entry~$X$ to an entry~$Y$ depicts the proper inclusion $X\subset Y$;
if two families are not connected by a directed path, then they are incomparable.
\end{theorem}

For proofs and references to proofs of the relations, we refer to \cite{Tru18-TRsubreg}.

\subsection{Contextual grammars}

Let $\cF$ be a family of languages. A contextual grammar with selection in $\cF$ is a triple
\[G=(V,\cS,A)\]
where
\begin{itemize}
\item $V$ is an alphabet, 
\item $\cS$ is a finite set of selection pairs $(S,C)$ where $S$ is a selection language over some subset $U$ of the
alphabet $V$ which belongs to the family $\cF$ with respect to the alphabet $U$ and where $C\subset V^*\times V^*$ 
is a finite set of contexts with the condition, for each context $(u,v)\in C$, at least one side is not empty: $uv\not=\lambda$,
\item $A$ is a finite subset of $V^*$ (its elements are called axioms).
\end{itemize}



Let $G=(V,\cS,A)$ be a contextual grammar with selection.
A direct internal derivation step in $G$ is defined as follows: a word~$x$ derives a word $y$ 
(written as $x\Lra y$) if and only if there are words $x_1$,~$x_2$,~$x_3$
with~$x_1x_2x_3=x$ and there is a selection pair~$(S,C)\in\cS$ such that $x_2\in S$ and~$y=x_1ux_2vx_3$ 
for some pair~$(u,v)\in C$.
Intuitively, we can only wrap a context~$(u,v)\in C$ around a subword $x_2$ of $x$ if $x_2$ belongs to 
the corresponding selection language~$S$.


By $\Lra^*$, we denote the reflexive and transitive closure of the relation~$\Lra$. 
The language generated by $G$ is defined as
\[L=\set{ z }{ x\Lra^* z \mbox{ for some } x\in A }.\]
By~$\cIC(\cF)$, we denote the family of all languages generated internally by contextual grammars 
with selection in $\cF$. When a contextual grammar works in the internal mode, we call it an internal 
contextual grammar. 

From previous research, we have the two hierarchies depicted in Figure~\ref{ic-fig-sep}.
An arrow from an entry~$X$ to an entry~$Y$ depicts the proper inclusion $X\subset Y$; a solid arrow
indicates that the inclusion is proper, the dashed arrow from $\cIC(\ORD)$ to $\cIC(\NC)$ indicates that it
is not known so far whether this inclusion is proper or whether equality holds.
The label at an edge shows in which paper the relation was proved. 

\begin{figure}[htb]
\centerline{%
\scalebox{0.83}{\begin{tikzpicture}[node distance=24.95mm and 25mm,on grid=true,
background rectangle/.style=
{
draw=black!80,
rounded corners=1ex},
show background rectangle]
\node (REG) {$\cIC(\REG)\stackrel{\text{\cite{DasManTru12b}}}{=}\cIC(\UF)$};
\node (PS) [below=of REG] {$\cIC(\PS)$};
\node (NC) [below=of PS] {$\cIC(\NC)$};
\node (ORD) [below=of NC] {$\cIC(\ORD)$};
\node (DEF) [below=of ORD] {$\cIC(\DEF)$};
\node (COMB) [below=of DEF] {$\cIC(\COMB)$};
\node (NIL) [left=of COMB] {$\cIC(\NIL)$};
\node (MON) [below=of COMB] {$\cIC(\MON)$};
\node (FIN) [left=of MON] {$\cIC(\FIN)$};
\node (SUF) [right=of ORD] {$\cIC(\SUF)$};
\node (COMM) [right=of SUF] {$\cIC(\COMM)$};
\node (CIRC) [above=of COMM] {$\cIC(\CIRC)$};
\draw[hier] (FIN) to node[edgeLabel]{\cite{DasManTru12b}} (NIL);
\draw[hier] (MON) to node[edgeLabel]{\cite{DasManTru12b}} (COMB);
\draw[hier] (MON) to node[edgeLabel]{\cite{DasManTru12b}} (NIL);
\draw[hier] (MON) [bend right=25]to node[edgeLabel]{\cite{DasManTru12b}} (SUF);
\draw[hier] (NIL) to node[edgeLabel,near start]{\cite{DasManTru12b}} (DEF);
\draw[hier] (MON) [bend right=30]to node[edgeLabel]{\cite{DasManTru12b}} (COMM);
\draw[hier] (COMB) to node[edgeLabel]{\cite{DasManTru12b}} (DEF);
\draw[hiero] (ORD) to (NC);
\draw[hier] (DEF) to node[edgeLabel]{\cite{Tru14-ncma}} (ORD);
\draw[hier] (NC) to node[edgeLabel]{\cite{Tru21-fi}} (PS);
\draw[hier] (PS) to node[edgeLabel]{\cite{Tru21-fi}} (REG);
\draw[hier] (COMM) to node[edgeLabel]{\cite{DasManTru12b}} (CIRC);
\draw[hier] (CIRC) [bend right=15]to node[edgeLabel]{\cite{DasManTru12b}} (REG);
\draw[hier] (SUF) [bend right=20]to node[edgeLabel]{\cite{Tru21-fi}} (PS);
\end{tikzpicture}}
\scalebox{0.832}{\begin{tikzpicture}[node distance=15mm and 28mm,on grid=true,
background rectangle/.style=
{
draw=black!80,
rounded corners=1ex},
show background rectangle]
\node (REG) {$\cIC(\REG)$};
\node (V10) [below=of REG] {$\vdots$};
\node (V9) [below=of V10] {$\cIC(\RL_n^V)$};
\node (V8) [below=of V9] {};
\node (V7) [below=of V8] {$\cIC(\RL_{n-1}^V)$};
\node (V6) [below=of V7] {$\vdots$};
\node (V3) [below=of V6] {$\cIC(\RL_2^V)$};
\node (V2) [below=of V3] {};
\node (V1) [below=of V2] {$\cIC(\RL_1^V)$};
\node (P12) [below left=of REG] {};
\node (P11) [below=of P12] {$\vdots$};
\node (P10) [below=of P11] {$\cIC(\RL_{2n}^P)$};
\node (P9) [below=of P10] {$\cIC(\RL_{2n-1}^P)$};
\node (P8) [below=of P9] {$\cIC(\RL_{2n-2}^P)$};
\node (P7) [below=of P8] {$\vdots$};
\node (P4) [below=of P7] {$\cIC(\RL_4^P)$};
\node (P3) [below=of P4] {$\cIC(\RL_3^P)$};
\node (P2) [below=of P3] {$\cIC(\RL_2^P)$};
\node (P1) [below=of P2] {$\cIC(\RL_1^P)$};
\node (Z11) [below right=of REG] {};
\node (Z10) [below=of Z11] {$\vdots$};
\node (Z9) [below=of Z10] {$\cIC(\REG_n^Z)$};
\node (Z8) [below=of Z9] {};
\node (Z7) [below=of Z8] {$\cIC(\REG_{n-1}^Z)$};
\node (Z6) [below=of Z7] {$\vdots$};
\node (Z3) [below=of Z6] {$\cIC(\REG_2^Z)$};
\node (Z2) [below=of Z3] {};
\node (Z1) [below=of Z2] {$\cIC(\REG_1^Z)$};
\draw[hier] (P1) to node[edgeLabel,pos=.4]{\cite{Tru21-fi}} (P2);
\draw[hier] (P2) to node[edgeLabel,pos=.4]{\cite{Tru21-fi}} (P3);
\draw[hier] (P3) to node[edgeLabel,pos=.4]{\cite{Tru21-fi}} (P4);
\draw[hier] (P4) to node[edgeLabel,pos=.4]{\cite{Tru21-fi}} (P7);
\draw[hier] (P7) to node[edgeLabel,pos=.4]{\cite{Tru21-fi}} (P8);
\draw[hier] (P8) to node[edgeLabel,pos=.4]{\cite{Tru21-fi}} (P9);
\draw[hier] (P9) to node[edgeLabel,pos=.4]{\cite{Tru21-fi}} (P10);
\draw[hier] (P10) to node[edgeLabel,pos=.4]{\cite{Tru21-fi}} (P11);
\draw[hier] (P11) [bend right=-22] to node[edgeLabel]{\cite{Tru21-fi}} (REG);
\draw[hier] (V1) to node[edgeLabel]{\cite{Tru21-fi}} (V3);
\draw[hier] (V3) to node[edgeLabel,pos=.4]{\cite{Tru21-fi}} (V6);
\draw[hier] (V6) to node[edgeLabel,pos=.4]{\cite{Tru21-fi}} (V7);
\draw[hier] (V7) to node[edgeLabel]{\cite{Tru21-fi}} (V9);
\draw[hier] (V9) to node[edgeLabel,pos=.4]{\cite{Tru21-fi}} (V10);
\draw[hier] (V10) to node[edgeLabel,pos=.4]{\cite{Tru21-fi}} (REG);
\draw[hier] (Z1) to node[edgeLabel]{\cite{Tru21-fi}} (Z3);
\draw[hier] (Z3) to node[edgeLabel,pos=.4]{\cite{Tru21-fi}} (Z6);
\draw[hier] (Z6) to node[edgeLabel,pos=.4]{\cite{Tru21-fi}} (Z7);
\draw[hier] (Z7) to node[edgeLabel]{\cite{Tru21-fi}} (Z9);
\draw[hier] (Z9) to node[edgeLabel,pos=.4]{\cite{Tru21-fi}} (Z10);
\draw[hier] (Z10) [bend right=22] to node[edgeLabel]{\cite{Tru21-fi}} (REG);
\draw[hier] (P2) to node[edgeLabel,pos=.4]{\cite{Tru21-fi}} (V1);
\draw[hier] (P4) to node[edgeLabel,pos=.4]{\cite{Tru21-fi}} (V3);
\draw[hier] (P8) to node[edgeLabel,pos=.4]{\cite{Tru21-fi}} (V7);
\draw[hier] (P10) to node[edgeLabel,pos=.4]{\cite{Tru21-fi}} (V9);
\draw[hier] (Z1) to node[edgeLabel,pos=.4]{\cite{Tru21-fi}} (V1);
\draw[hier] (Z3) to node[edgeLabel,pos=.4]{\cite{Tru21-fi}} (V3);
\draw[hier] (Z7) to node[edgeLabel,pos=.4]{\cite{Tru21-fi}} (V7);
\draw[hier] (Z9) to node[edgeLabel,pos=.4]{\cite{Tru21-fi}} (V9);
%
\end{tikzpicture}}
}
\caption{Hierarchies of the language families by internal contextual grammars with
selection languages defined by structural properties (left) or restricted resources (right). 
An edge label refers to the paper where the respective inclusion is proved.}\label{ic-fig-sep}
\end{figure}
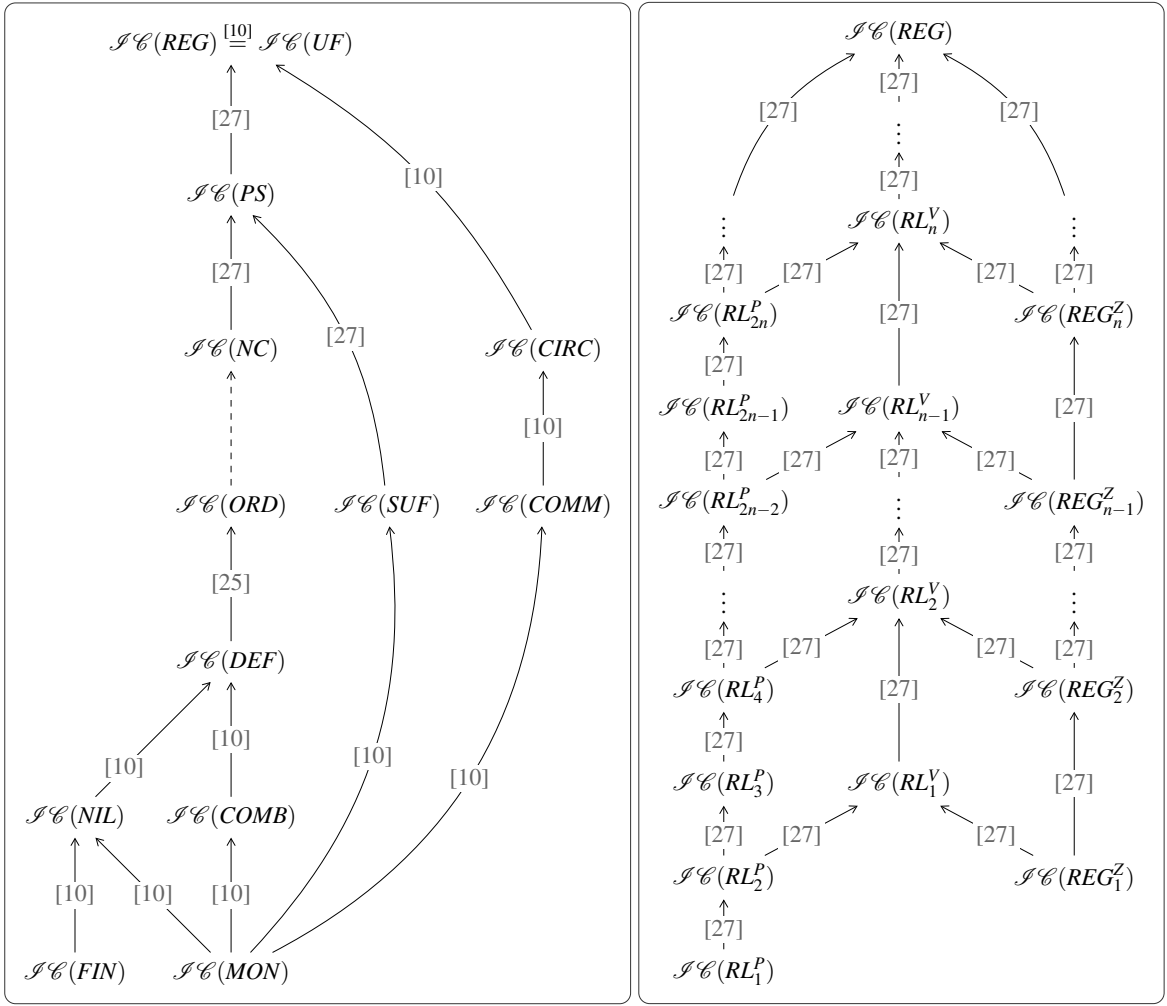

If two families~$X$ and~$Y$ are not connected by a directed path, then~$X$ and~$Y$ are
in most cases incomparable. The only exceptions are the relations of the family $\cIC(\SUF)$
to the families $\cIC(\ORD)$ and~$\cIC(\NC)$ where it is not known whether they are incomparable
or whether $\cIC(\SUF)$ is a subset of the other and the relation of the family $\cIC(REG_{n+1}^Z)$
to $\cIC(RL_n^V)$ for $n\geq 1$ where it is not known whether they are incomparable
or whether $\cIC(REG_{n+1}^Z)$ is a subset of $\cIC(RL_n^V)$. 

We note here that in \cite{Das05, DasManTru11b, DasManTru12a, DasManTru12b, ManTru12, Tru14-ncma, Das-Analele15}
a slightly different definition was used than in \cite{Tru21-fi, DasTru22-ncma} and the present
paper. This difference consists in the alphabet of the selection languages. In the early papers,
the selection languages belong to some subfamily $\cF$ with respect to the whole alphabet $V$ of the
contextual grammar whereas in later papers, the selection languages belong to some subfamily $\cF$
with respect to some subalphabet $U\subseteq V$ of the contextual grammar. The language $\{a\}^*\{a\}^5$,
for instance, is nilpotent with respect to the alphabet $\{a\}$ but not with respect to the alphabet $\{a,b\}$.
For almost all proofs in the mentioned papers, there is no difference between using one or the other definition.
The only proof which relies on the definition is that of the relation $L(G)\notin\cIC(\DEF)$ for
\[G=(V,\sets{(\Suf(\sets{d}^*\sets{b}),\sets{(a,b)}),(\sets{a,\lambda},\sets{(c,d)})},\sets{ecadb})\]
from \cite[Lemma 21]{DasManTru12b} (also used in \cite[Theorem~3.5]{Tru14-ncma}).
However, the proof is valid also with the subalphabet definition if one changes the axiom $ecadb$ to the word $dcadb$.

From the definition follows that the subset relation is preserved under the use of contextual
grammars: if we allow more, we do not obtain less.

\begin{lemma}\label{l-context-gramm-monoton}
For any two language classes $X$ and $Y$ with $X\subseteq Y$,
we have the inclusion
\[\cIC(X)\subseteq\cIC(Y).\]
\end{lemma}

\noindent In the following section, we relate the families of the two hierarchies mentioned above.
\pagebreak

\section{Results}

When we speak about contextual grammars in this section, we mean internal contextual grammars (whose languages
are generated in the internal mode).

First, we present languages which will serve later as witness languages for proper inclusions or incomparabilities.

\begin{lemma}\label{ic-comp-l1}
Let $V=\{a,b,c,d,e\}$ be an alphabet, $G=(V,\sets{(S_1,C_1),(S_2,C_2)},\{c\})$ be a contextual grammar
with
\[\begin{array}{rcl@{\quad}rcl}
S_1 &=& \sets{b}^*\sets{c},& C_1 &=& \sets{(ab,ab)},\\
S_2 &=& \sets{aa}^*, & C_2 &=& \sets{(d,e)},
\end{array}\]
and $L=L(G)$ be the laguage generated.
Then 
\[L\in(\cIC(\RL_1^V)\cap\cIC(\RL_2^P)\cap\cIC(\REG_2^Z))\setminus\cIC(\PS).\]
\end{lemma}
\begin{proof}
The selection languages are generated by right-linear grammars with the following rules (and start symbol $S$):
\begin{align*}
& S_1:\  S\ra bS,\ S\ra c,\\
& S_2:\  S\ra aaS,\ S\ra \lambda.
\end{align*}
Since these rules contain one non-terminal symbol only and two rules each, we obtain 
\[L\in\cIC(\RL_1^V)\cap\cIC(\RL_2^P).\]

Since the words of the language $L$ contain only one letter $c$ (the axiom has no more and the contexts do not contain $c$),
the language $L$ is also generated if $S_1$ is replaced by the language $S_1'=(\{b\}^*\{c\})^+$ (the additional words
cannot be used for selection).

The selection languages $S_1'$ and $S_2$ are accepted by automata with two states each 
whose transition functions are given in the following diagram:
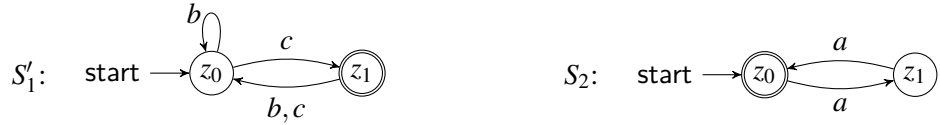
\begin{center}
\raisebox{6.5mm}{$S_1'$:}\quad
\begin{picture}(45,18)
\put(0,0){%
\begin{tikzpicture}[on grid,>=stealth',initial text={\sf start}
]
\node[state,minimum size=4mm,initial] (z0) at (0,0) {$z_0$};
\node[state,minimum size=4mm,accepting] (z1) at (2,0) {$z_1$};
\draw[->] (z0) [bend left=15]to node [above,sloped] {$c$} (z1);
\draw[->] (z0) .. controls +(75:1) and +(105:1) .. node [left] {$b$} (z0);
\draw[->] (z1) [bend left=15]to node [below,sloped] {$b,c$} (z0);
\end{tikzpicture}}
\end{picture}\hspace*{20mm}%
\raisebox{6.5mm}{$S_2$:}\quad
\begin{picture}(45,18)
\put(0,1.5){%
\begin{tikzpicture}[on grid,>=stealth',initial text={\sf start}
]
\node[state,minimum size=4mm,initial,accepting] (z0) at (0,0) {$z_0$};
\node[state,minimum size=4mm] (z1) at (2,0) {$z_1$};
\draw[->] (z0) [bend right=15]to node [below,sloped] {$a$} (z1);
\draw[->] (z1) [bend right=15]to node [above,sloped] {$a$} (z0);
\end{tikzpicture}}
\end{picture}
\end{center}
Hence, $L\in\cIC(\REG_2^Z)$.

Now we prove that $L$ cannot be generated by a contextual grammar where all selection languages are
power-separating.
Assume the contrary. Then there is a contextual grammar $G'=(V,\cS,A)$ which also generates the language $L$ 
and all selection languages belong to the class $\PS$. For each selection language~$S$ occurring in $\cS$, 
there exists a natural number $m_S\geq 1$
such that for all words $x\in V^*$ it holds 
\[\mbox{either } J_x^{m_S}\cap S=\emptyset \mbox{ or } J_x^{m_S}\subseteq S \mbox{ with } J_x^{m_S}=\set{x^n}{n\geq {m_S}}.\]
Since $\cS$ is finite, there is also a natural number $m\geq 1$ such that 
\[\mbox{either } J_x^m\cap S=\emptyset\mbox{ or } J_x^m\subseteq S \mbox{ with } J_x^m=\set{x^n}{n\geq m}\]
holds for every selection language $S$.
Now let $m$ be such a value.

Further, let $k$ be the maximal length of the axioms and contexts plus $m$:
\[k=\max\sets{\max\set{|w|}{w\in A},\max\set{|uv|}{(u,v)\in C, (S,C)\in\cS}}+m.\]

Consider the word $w=da^{2k}eb^{2k}c(ab)^{2k}$ which belongs to the language $L$ but not to the set $A$ of axioms
due to its length.
Therefore, it is derived from another word $w'\in L$ by insertion of a context $(u,v)$ from a selection pair~$(S,C)$.
We now study the possibilities for $u$ and, depending from this, also for $v$. Let $w'_1$, $w'_2$, and $w'_3$ be the
subwords of $w'$ which are separated by the insertion of $(u,v)$: 
\[w' = w'_1w'_2w'_3\Lra w'_1uw'_2vw'_3 =w.
\]

If $u=d$, then $v=e$. This case will be continued later.

If $u=da^n$ for some $n$ with $1\leq n\leq k$, then $v$ contains the letter $e$ and has to bear also $n$ letters $b$ to
be inserted before the letter $c$ but also $n$ letters of $a$ and $b$ to be created after the $c$ which is not possible.

If $u=a^n$ for some number $n$ with $1\leq n\leq k$ and $w'_1=da^p$ for some $p$ with $0\leq p\leq 2k-n$, 
then $v$ has to bear also $n$ letters $b$ to be inserted before the letter $c$ and also $n$ letters of $a$ and $b$ 
to be created after the $c$ which is not possible.

It is not possible that $u$ contains the letter $e$ because $d$ and $e$ are inserted at the same time but $d$ cannot be present 
in $u$ together with $e$ due to the length of $u$. 

If $w'_1$ starts with $da^{2k}e$ (if $u$ as a subword of $w$ starts after the letter $e$), 
then the word $w'$ does not have the correct form (does not belong to the language $L$ which is a contradiction), 
since the number of letters~$a$ before $c$ is already $2k$ whereas the number of occurrences of $b$ before $c$ or 
the number of occurrences of~$ab$ after $c$ is less (since $|uv|>0$).

Thus, the only possibility is that $(u,v)=(d,e)$ and $w'_2=a^{2k}$. We have $2k>m$ and, therefore, $a^{2k}\in J_a^m$.
Hence, $J_a^m\cap S\not=\emptyset$ and $J_a^m\subseteq S$. Therefore, the word $a^{2k+1}$ (which belongs to the
set $J_a^m$) also belongs to the selection language $S$. The language $L$ also contains the 
word $a^{2k+1}b^{2k+1}c(ab)^{2k+1}$. With the same selection pair $(S,C)$, the word $da^{2k+1}eb^{2k+1}c(ab)^{2k+1}$
could be derived. But this does not belong to the language $L$.
This contradiction shows that our assumption was wrong and that $L\notin\cIC(\PS)$ holds.
\end{proof}
\begin{lemma}\label{ic-comp-l2}
Let $L=\Set{c^nac^mbc^{n+m}}{n\geq 0,m\geq 0}\cup\Set{c^nbc^na}{n\geq 0}$.
Then the relation
\[L\in(\cIC(\RL^V_1)\cap\cIC(\RL^P_2)\cap\cIC(\REG^Z_2))\setminus(\cIC(\CIRC)\cup\cIC(\SUF))\]
holds.
\end{lemma}
\begin{proof}
Let $V=\Sets{a,b,c}$.
The language $L$ is generated by the contextual grammar
\[G=(V,\sets{(\sets{ab,b},\sets{(c,c)})},\sets{ab,ba}).\]
Since the selection language is finite with two words, it can be generated by a right-linear grammar
with one non-terminal symbol and two rules only. Hence, $L\in\cIC(\RL^V_1)\cap\cIC(\RL^P_2)$.

The language $L$ is also generated by the contextual grammar
\[G=(V,\sets{(V^*\sets{b},\sets{(c,c)})},\sets{ab,ba}).\]
with a combinational selection language only. Every combinational language is accepted by a deterministic finite
automaton with two states (see Theorem~\ref{th-subreg-hier} and Figure~\ref{fig-subreg-hier}).
Hence, $L\in\cIC(\REG^Z_2)$.

In \cite[Lemma~18]{DasManTru12b}, it was shown that the language $L$ can neither be generated by a
contextual grammar with circular filters nor by one with suffix-closed filters. Hence, $L\notin\cIC(\CIRC)\cup\cIC(\SUF)$.
\end{proof}
\begin{lemma}\label{ic-comp-l3}
Let $n\geq 1$ be a natural number and let
\[A_n=\sets{a_1,\ldots,a_n},\quad B_n=\sets{b_1,\ldots,b_n},\quad C_n=\sets{c_1,\ldots,c_n},\quad D_n=\sets{d_1,\ldots,d_n},\]
as well as
\begin{align*}
V_n&=A_n\cup B_n\cup C_n\cup D_n,\\
P_n&=\set{(a_i,c_j)}{1\leq i\leq n, 1\leq j\leq n},\\
Q_n&=\set{(b_i,d_j)}{1\leq i\leq n, 1\leq j\leq n},\\
G_n&=(V_n,\sets{(B_n^*,P_n),(C_n^*,Q_n)},\set{a_{i_a}b_{i_b}c_{i_c}d_{i_d}}{1\leq i_x\leq n,\ x\in\sets{a,b,c,d}}),
\end{align*}
and $L_n=L(G_n)$. 
Then the relation
$L_n\in\cIC(\MON)\setminus\cIC(\RL^P_n)$
holds.
\end{lemma}
\begin{proof}
Let $n\geq 1$. The selection languages of $G_n$ are monoidal. Thus, $L_n\in\cIC(\MON)$.\\
From \cite[Lemma~3.30]{Tru21-fi}, we know that $L\notin\cIC(\RL^P_n)$.
%
\end{proof}

\begin{lemma}\label{ic-comp-l4}
Let $V=\sets{a,b}$ and
$L_n=\set{a^{p_0}ba^{p_1}ba^{p_2}b\cdots a^{p_n}ba^{p_0}ba^{p_1}ba^{p_2}b\cdots a^{p_n}}{p_i\geq 1,\ 0\leq i\leq n}$
for~$n\geq 1$. Then
\[L_n\in(\cIC(\COMM)\cap\cIC(\ORD))\setminus\cIC(\RL^V_n).\]
\end{lemma}
\begin{proof}
Let $n$ be a natural number with $n\geq 1$.

The language $L_n$ is generated by the contextual grammar
\[G_n=(V,\sets{(S_n,\sets{(a,a)})},\sets{(ab)^{2n+1}a})\]
with the selection language $S_n=(\sets{a}^*\sets{b}\sets{a}^*)^{n+1}$.
This selection language is commutative; hence, we have $L_n\in\cIC(\COMM)$.

The selection language is accepted by an automaton whose transition function is shown in the following diagram:
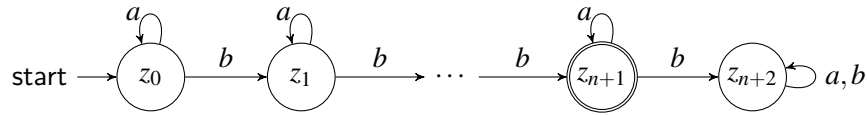
\begin{center}
\begin{tikzpicture}[on grid,>=stealth',initial text={\sf start}
]
\node[state,minimum size=9mm,initial] (z0) at (0,0) {$z_0$};
\node[state,minimum size=9mm] (z1) at (2,0) {$z_1$};
\node (z2) at (4,0) {$\cdots$};
\node[state,minimum size=9mm,accepting] (z3) at (6,0) {$z_{n+1}$};
\node[state,minimum size=9mm] (z4) at (8,0) {$z_{n+2}$};
\draw[->] (z0) .. controls +(75:1) and +(105:1) .. node [left] {$a$} (z0);
\draw[->] (z0) edge node [above,sloped] {$b$} (z1);
\draw[->] (z1) .. controls +(75:1) and +(105:1) .. node [left] {$a$} (z1);
\draw[->] (z1) edge node [above,sloped] {$b$} (z2);
\draw[->] (z2) edge node [above,sloped] {$b$} (z3);
\draw[->] (z3) .. controls +(75:1) and +(105:1) .. node [left] {$a$} (z3);
\draw[->] (z3) edge node [above,sloped] {$b$} (z4);
\draw[->] (z4) .. controls +(345:1) and +(15:1) .. node [right] {$a,b$} (z4);
\end{tikzpicture}
\end{center}

This shows that the automaton is ordered (with $z_0\prec z_1\prec \cdots \prec z_{n+2}$, it holds
$\delta(z_i,x)\preceq\delta(z_j,x)$ for any two states $z_i$ and $z_j$ with $z_i\prec z_j$ and any $x\in\{a,b\}$).
Hence, $L_n\in\cIC(\ORD)$.\\
In \cite[Lemma~3.29]{Tru21-fi}, the relation $L_n\notin\cIC(\RL^V_n)$ was proved.
\end{proof}


\begin{lemma}\label{ic-comp-l6}
Let $n\geq 2$ be a natural number, $V_n=\{a_1,a_2,\ldots,a_n\}$ be an alphabet, and $L_n$ be the 
language~$L_n=\sets{a_1a_2\ldots a_n}^+\cup V_n^{n-1}$.
Then the relation
$L_n\in\cIC(\FIN)\setminus\cIC(\REG^Z_n)$
holds.
\end{lemma}
\begin{proof}
Let $n\geq 2$.
The language $L_n$ is generated by the contextual grammar
\[G_n=(V_n,\sets{(\{a_1a_2\ldots a_n\},\sets{(\lambda,a_1a_2\ldots a_n)})},V_n^{n-1}\cup\sets{a_1a_2\ldots a_n})\]
with a finite selection language only. Thus, $L_n\in\cIC(\FIN)$.\\
In \cite[Lemma~3.31]{Tru21-fi}, it was shown that $L_n\notin\cIC(\REG^Z_n)$.
\end{proof}

In a similar way, the following result is proved.

\begin{lemma}\label{ic-comp-l7}
Let $n\geq 2$ be a natural number, $V_n=\{a_1,a_2,\ldots,a_n\}$ be an alphabet, and $L_n$ be the language
\[L_n=V_n^{{}\leq n-1}\cup\bigcup_{k\geq 1}V_n^{kn}.\]
Then the relation
$L_n\in\cIC(\COMM)\setminus\cIC(\REG^Z_n)$
holds.
\end{lemma}
\begin{proof}
Let $n\geq 2$.
The language $L_n$ is generated by the contextual grammar
\[G_n=(V_n,\sets{(V_n^n,\set{(\lambda,w)}{w\in V_n^n})},V_n^{{}\leq n-1}\cup V_n^n)\]
with a commutative selection language only. Thus, $L_n\in\cIC(\COMM)$.

In any contextual grammar generating the language $L_n$, every context has a length which is divisible
by $n$ and can only be added to subwords of words of the language which have a length of at least $n$. Since
every subword of length less than $n$ occurs in the language, the selected subwords must have a length
of at least $n$. This cannot be checked by a deterministic finite automaton with $n$ states only.
\end{proof}


We now prove the relations between the language families of contextual grammars where the 
selection languages are taken from subregular families of languages which have common 
structural properties and from families of regular languages defined by restricting the 
resources needed for generating or accepting them.
%
We start with families which are defined by the number of production rules necessary for 
generating the selection languages.


\begin{lemma}\label{ic-l-P1-FIN}
The language families $\cIC(\RL^P_1)$ and $\cIC(\FIN)$ coincide.
\end{lemma}
\begin{proof}
The inclusion $\cIC(\RL^P_1)\subseteq\cIC(\FIN)$ follows by Lemma~\ref{l-context-gramm-monoton} from 
the inclusion $\RL^P_1\subseteq\FIN$ (see Theorem~\ref{th-subreg-hier} and also Figure~\ref{fig-subreg-hier}).

For the converse inclusion,
let $m\geq 1$ and 
\[G=(V,\set{(S_i,C_i)}{1\leq i\leq m},A)\]
be a contextual grammar where all 
selection languages $S_i$ ($1\leq i\leq m$) are finite. Then we split up the selection languages into singleton sets
and obtain the contextual grammar
\[G'=(V,\set{(\sets{w},C_i)}{1\leq i\leq m,\ w\in S_i},A)\]
which generates the same language as $G$
and all selection languages belong to the family $\RL^P_1$. Hence, also the inclusion $\cIC(\FIN)\subseteq\cIC(\RL^P_1)$
holds and together we obtain $\cIC(\FIN)=\cIC(\RL^P_1)$
\end{proof}

\begin{lemma}\label{ic-incomp-P2-mon-ps-circ}
The language families $\cIC(\RL^P_n)$ for $n\geq 2$ are incomparable to the 
families 
\begin{multline*}
\cIC(\MON),\ \cIC(\NIL),\ \cIC(\COMB),\ \cIC(\DEF),\ \cIC(\ORD),\ 
\cIC(\NC),\ \cIC(\PS),\\ \cIC(\SUF),\ \cIC(\COMM), \text{ and }\cIC(\CIRC).
\end{multline*}
\end{lemma}
\begin{proof}
Due to the inclusion relations stated in Theorem~\ref{th-subreg-hier}, depicted in
Figure~\ref{fig-subreg-hier}, proofs of the following relations are sufficient:
\begin{enumerate}
\item $\cIC(\RL^P_2)\setminus\cIC(\PS)\not=\emptyset$,
\item $\cIC(\RL^P_2)\setminus\cIC(\CIRC)\not=\emptyset$,
\item $\cIC(\MON)\setminus\cIC(\RL^P_n)\not=\emptyset$ for every natural number $n$ 
with $n\geq 2$.
\end{enumerate}
The first relation was proved in Lemma~\ref{ic-comp-l1}, the second relation in Lemma~\ref{ic-comp-l2},
and the third relation in Lemma~\ref{ic-comp-l3}.
\end{proof}

Regarding the families which are defined by the number of states necessary for accepting
the selection languages, we obtain the following results.

\begin{lemma}\label{ic-comp-mon-Z1}
The language families $\cIC(\MON)$ and $\cIC(\REG^Z_1)$ coincide.
\end{lemma}
\begin{proof}
This follows from the fact that $\REG^Z_1=\MON\cup\sets{\emptyset}$
and that the empty set has no influence as a selection language.
\end{proof}

\begin{lemma}\label{ic-comp-comb-Z2}
The relation
$\cIC(\COMB)\subset\cIC(\REG^Z_2)$
holds.
\end{lemma}
\begin{proof}
From 
Theorem~\ref{th-subreg-hier} (see Figure~\ref{fig-subreg-hier}), we know that $\COMB\subset\REG^Z_2$. By
Lemma~\ref{l-context-gramm-monoton}, we obtain that~$\cIC(\COMB)\subseteq\cIC(\REG^Z_2)$
holds. 
By Lemma~\ref{ic-comp-l1}, this inclusion is proper.
\end{proof}

\begin{lemma}\label{ic-incomp-Z2-fin-ps}
Every language family $\cIC(\REG^Z_n)$ where $n\geq 2$ is incomparable to each of the 
families
\[\cIC(\FIN),\ \cIC(\NIL),\ \cIC(\DEF),\ \cIC(\ORD),\ \cIC(\NC), \text{ and }\cIC(\PS).\]
\end{lemma}
\begin{proof}
Due to the inclusion relations stated in Theorem~\ref{th-subreg-hier}, depicted in
Figure~\ref{fig-subreg-hier}, proofs of the following relations are sufficient:
\begin{enumerate}
\item $\cIC(\REG^Z_2)\setminus\cIC(\PS)\not=\emptyset$,
\item $\cIC(\FIN)\setminus\cIC(\REG^Z_n)\not=\emptyset$ for every $n\geq 2$.
\end{enumerate}
The first relation was proved with Lemma~\ref{ic-comp-l1}, the second one with Lemma~\ref{ic-comp-l6}.
\end{proof}

\begin{lemma}\label{ic-incomp-Z2-comm-circ}
Every language family $\cIC(\REG^Z_n)$ where $n\geq 2$ is incomparable to each of the 
fa\-mi\-lies~$\cIC(\COMM)$ and $\cIC(\CIRC)$.
\end{lemma}
\begin{proof}
Due to the inclusion relations stated in Theorem~\ref{th-subreg-hier}, depicted in
Figure~\ref{fig-subreg-hier}, proofs of the following relations are sufficient:
\begin{enumerate}
\item $\cIC(\REG^Z_2)\setminus\cIC(\CIRC)\not=\emptyset$,
\item $\cIC(\COMM)\setminus\cIC(\REG^Z_n)\not=\emptyset$ for every $n\geq 2$.
\end{enumerate}
The first relation was proved with Lemma~\ref{ic-comp-l2}, the second one with Lemma~\ref{ic-comp-l7}.
\end{proof}


Regarding the families which are defined by the number of non-terminal symbols necessary for 
generating the selection languages, we obtain the following results.

\begin{lemma}\label{ic-comp-def-V1}
The relation
$\cIC(\DEF)\subset\cIC(\RL^V_1)$
holds.
\end{lemma}
\begin{proof}
We first prove the inclusion $\cIC(\DEF)\subseteq\cIC(\RL^V_1)$.

Let $n\geq 1$ and 
\[G=(V,\set{(S_i,C_i)}{1\leq i\leq n},A)\]
be a contextual grammar where every selection language can be represented in the form
$S_i=A_i\cup V^*B_i$ with $1\leq i\leq n$ for finite subsets $A_i$ and $B_i$ of $V^*$.
The same language $L(G)$ is also generated by the contextual grammar
\[G'=(V,\set{(A_i,C_i)}{1\leq i\leq n}\cup\set{(V^*B_i,C_i)}{1\leq i\leq n},A).\]
Every such selection language 
$A_i$ and $V^*B_i$ for $1\leq i\leq n$ 
can be generated by a right-linear grammar with one non-terminal symbol only:
\[G_{A_i}=(\sets{S},V,\set{S\ra w}{w\in A_i},S)\]
for generating the language $A_i$ and
\[G_{B_i}=(\sets{S},V,\set{S\ra xS}{x\in V}\cup\set{S\ra w}{w\in B_i},S)\]
for generating the language $V^*B_i$.
Hence, $\cIC(\DEF)\subseteq\cIC(\RL^V_1)$.

With Lemma~\ref{ic-comp-l1}, it is proved that a language exists in the set $\cIC(\RL^V_1)\setminus\cIC(\PS)$.
This language is also a witness language for the properness of the inclusion $\cIC(\DEF)\subset\cIC(\RL^V_1)$.
\end{proof}

\begin{lemma}\label{ic-incomp-V1-ord-ps-comm-circ}
Every language family $\cIC(\RL^V_n)$ where $n\geq 1$ is incomparable to the 
families 
\[\cIC(\ORD),\ \cIC(\NC),\ \cIC(\PS),\ \cIC(\COMM), \text{ and }\cIC(\CIRC).\] 
\end{lemma}
\begin{proof}
Due to the inclusion relations stated in Theorem~\ref{th-subreg-hier}, depicted in 
Figure~\ref{fig-subreg-hier}, proofs of the following relations are sufficient: 
\begin{enumerate}
\item $\cIC(\RL^V_1)\setminus\cIC(\PS)\not=\emptyset$,
\item $\cIC(\RL^V_1)\setminus\cIC(\CIRC)\not=\emptyset$,
\item $\cIC(\COMM)\setminus\cIC(\RL^V_n)\not=\emptyset$ for every $n\geq 1$,
\item $\cIC(\ORD)\setminus\cIC(\RL^V_n)\not=\emptyset$ for every $n\geq 1$.
\end{enumerate}
The first relation is proved in Lemma~\ref{ic-comp-l1},
the second in Lemma~\ref{ic-comp-l2}, and the other two in Lemma~\ref{ic-comp-l4}.
\end{proof}


The following theorem summarizes the results.

\begin{theorem}\label{th-ec-comp}
The relations depicted in Figure~\ref{ec-fig-comp} hold.
An arrow from an entry~$X$ to an entry $Y$ denotes the proper inclusion $X\subset Y$. 
If two families are not connected by a directed path then they are not necessarily incomparable.
\end{theorem}

\begin{figure}[htb]
\centerline{%
\scalebox{0.75}{\begin{tikzpicture}[node distance=19mm and 25mm,on grid=true,
background rectangle/.style=
{
draw=black!80,
rounded corners=1ex},
show background rectangle]
\node (REG) {$\cIC(\REG)\stackrel{\text{\cite{DasManTru12b}}}{=}\cIC(\UF)$};
\node (PS) [below=of REG] {$\cIC(\PS)$};
\node (NC) [below=of PS] {$\cIC(\NC)$};
\node (ORD) [below=of NC] {$\cIC(\ORD)$};
\node (DEF) [below=of ORD] {$\cIC(\DEF)$};
\node (dummy2) [left=of DEF] {};
\node (V1) [above right=of DEF] {$\cIC(\RL_1^V)$};
\node (V2) [above=of V1] {$\vdots$};
\node (COMB) [below=of DEF] {$\cIC(\COMB)$};
\node (Z2) [right=of V1] {$\cIC(\REG_2^Z)$};
\node (Z3) [above=of Z2] {$\vdots$};
\node (NIL) [left=of COMB] {$\cIC(\NIL)$};
\node (dummy) [below=of COMB] {};
\node (MON) [below=of dummy] {\hspace*{2cm}$\cIC(\MON)\stackrel{\text{\footnotesize\ref{ic-comp-mon-Z1}}}{=}\cIC(\REG_1^Z)$};
\node (FIN) [left=of dummy] {\hspace*{-10mm}$\cIC(\RL_1^P)\stackrel{\text{\footnotesize\ref{ic-l-P1-FIN}}}{=}\cIC(\FIN)$};
\node (SUF) [left=of dummy2] {$\cIC(\SUF)$};
\node (CIRC) [right=of Z2] {$\cIC(\CIRC)$};
\node (COMM) [below=of CIRC] {$\cIC(\COMM)$};
\draw[hier] (FIN) to 
  (NIL);
\draw[hier] (MON) to 
  (COMB);
\draw[hier] (MON) [bend right=25]to 
  (NIL);
\draw[hier,rounded corners] (MON) |- +(-1,0.5) -| (SUF);
\draw[hier] (NIL) to 
  (DEF);
\draw[hier] (MON) [bend right=30]to 
  (COMM);
\draw[hier] (COMB) to 
  (DEF);
\draw[hier] (COMB) to [bend right=22] node[edgeLabel]{\ref{ic-comp-comb-Z2}} 
  (Z2);  
\draw[hiero] (ORD) to (NC);
\draw[hier] (DEF) to 
  (ORD);
\draw[hier] (DEF) to [bend right=22] node[edgeLabel]{\ref{ic-comp-def-V1}} 
  (V1);
\draw[hier] (NC) to 
  (PS);
\draw[hier] (PS) to 
  (REG);
\draw[hier] (COMM) to 
  (CIRC);
\draw[hier] (CIRC) [bend right=29]to 
  (REG);
\draw[hier] (SUF) [bend left=20]to 
  (PS);
\draw[hier] (V1) to (V2);
\draw[hier] (Z2) to (Z3);
\draw[hier] (V2) [bend right=22]to (REG);
\draw[hier] (Z3) [bend right=22]to (REG);
\end{tikzpicture}}}
\caption{Hierarchy of language families by contextual grammars; an edge label refers to the corresponding lemma
(where the relation was not already shown in Figure~\ref{ic-fig-sep}). The incomparabilities were proved in the 
Lemmas~\ref{ic-incomp-P2-mon-ps-circ},~\ref{ic-incomp-Z2-fin-ps},~\ref{ic-incomp-Z2-comm-circ}
and~\ref{ic-incomp-V1-ord-ps-comm-circ}.}%
\label{ec-fig-comp}
\end{figure}

If two families~$X$ and~$Y$ are not connected by a directed path, then~$X$ and~$Y$ are
in most cases incomparable. The only exceptions are the relations of the family $\cIC(\SUF)$
to the families $\cIC(\ORD)$ and~$\cIC(\NC)$, to the families $\cIC(\RL^V_n)$ for $n\geq 1$, to the 
families $\cIC(\REG^Z_n)$ for $n\geq 2$ where it is not known whether they are incomparable
or whether $\cIC(\SUF)$ is a subset of the other and the relation of the family $\cIC(REG_{n+1}^Z)$
to $\cIC(RL_n^V)$ for $n\geq 1$ where it is not known whether they are incomparable
or whether $\cIC(REG_{n+1}^Z)$ is a subset of $\cIC(RL_n^V)$. 

\pagebreak

\section{Conclusions and Further Work}

In \cite{Tru21-fi}, two independent hierarchies have been obtained for each type of contextual grammars, 
one based on selection languages defined by structural properties, the other one based on resources.
In the present paper, these hierarchies have been merged for internal contextual grammars.

Some questions remain open:
\begin{itemize}
\item Let $n\geq 1$. Is there a language $L_n\in\cIC(\SUF)\setminus\cIC(\RL^V_n)$?
\item Let $n\geq 2$. Is there a language $L_n\in\cIC(\SUF)\setminus\cIC(\REG^Z_n)$ for $n\geq 2$?
\end{itemize}

If the first question is answered affirmatively, then these languages $L_n$ satisfy also
$L_n\notin\cIC(\REG^Z_n)$ since $\cIC(\REG^Z_n)\subset\cIC(\RL^V_n)$ for $n\geq 1$ 
(Theorem~\ref{th-subreg-hier}, see Figure~\ref{fig-subreg-hier}).

If such languages are found, then it is clear that every language family
$\cIC(\RL^V_n)$ for $n\geq 1$ and every languages family 
$\cIC(\REG^Z_n)$ for $n\geq 2$ is incomparable to the family $\cIC(\SUF)$.
So far, we only know that $\cIC(\RL^V_n)\not\subseteq\cIC(\SUF)$ for $n\geq 1$
and that $\cIC(\REG^Z_n)\not\subseteq\cIC(\SUF)$ for $n\geq 2$ (both shown in Lemma~\ref{ic-comp-l1})

Recently, in \cite{Das-Analele15,DasTru22-ncma}, strictly locally $k$-testable languages have been
investigated as selection languages for contextual grammars. 
Also for the language families defined by those selection languages, it should be investigated where 
they are located in the presented hierarchy.

Additionally, other subfamilies of regular languages could be taken into consideration. Recently,
in~\cite{Das17,Das21}, external contextual grammars have been investigated where the selection 
languages are ideals or codes.
This reseach could be extended to internal contextual grammars with ideals or codes as selection languages.

\bibliographystyle{eptcs}
\bibliography{cont-subreg}

\end{document}